%% file: main_long.tex
\documentclass[letterpaper, 10 pt, conference]{ieeeconf}  

\IEEEoverridecommandlockouts                              

\usepackage{blindtext}
\usepackage{cite}
\usepackage{amsmath}
\usepackage{amssymb,amsfonts}
\usepackage{graphicx}
\usepackage{multicol}%
\usepackage{dblfloatfix}
\usepackage{hyperref}
\usepackage{grffile}
\usepackage{textcomp}
\usepackage{float}
\usepackage{booktabs}	
\usepackage{wrapfig}
\usepackage{blkarray}
\usepackage[capitalise]{cleveref}
\usepackage{centernot}

\newcommand{\nll}{\centernot{\ll}}

\usepackage{dsfont}
\usepackage{amsthm}
\usepackage[normalem]{ulem}
\newcommand\canceltext{\bgroup\markoverwith
{\textcolor{red}{\rule[.5ex]{2pt}{0.4pt}}}\ULon}

\crefformat{section}{\S#2#1#3}
\crefformat{subsection}{\S#2#1#3}
\crefformat{subsubsection}{\S#2#1#3}
\crefrangeformat{section}{\S\S#3#1#4 to~#5#2#6}
\crefmultiformat{section}{\S\S#2#1#3}{ and~#2#1#3}{, #2#1#3}{ and~#2#1#3}

 \usepackage[linesnumbered, ruled, vlined]{algorithm2e}
 \SetKwRepeat{Do}{do}{while}%
 \newcommand{\indi}{\mathds{1}}

\usepackage{url}
\usepackage{color, colortbl}
\definecolor{Gray}{gray}{0.95}
\definecolor{LightCyan}{rgb}{0.88,1,1}

\usepackage{bbm}
\usepackage{mathtools}
\usepackage{xcolor}
\usepackage{subcaption}
\newtheorem{assumption}{Assumption}
\newtheorem{lemma}{Lemma}
\newtheorem{theorem}{Theorem}
\newtheorem{proposition}{Proposition}

\newcommand{\set}{\mathcal}

\newcommand{\tr}{\mathrm{Tr}}

\DeclareMathOperator{\supp}{supp}

\DeclareMathOperator*{\esssup}{ess\,sup}
\newtheorem{definition}{Definition}

\newtheorem{remark}{Remark}

\let\oldref\ref
\renewcommand{\ref}[1]{(\oldref{#1})}

\SetKwComment{Comment}{$\triangleright$\ }{}
\SetCommentSty{itshape}
\title{\LARGE \bf
Minimizing Information Leakage of Abrupt Changes\\in Stochastic Systems
}

\author{Alessio Russo$^{\star,1}$ and Alexandre Proutiere$^{1}$
\thanks{$^\star$ Corresponding author}
\thanks{$^{1}$Alessio Russo and Alexandre Proutiere are in the Division of Decision and Control Systems of the EECS School at KTH Royal Institute of Technology, Stockholm, Sweden.
        {\tt\small \{alessior,alepro\}@kth.se}}
}

\begin{document}

\maketitle
\thispagestyle{empty}
\pagestyle{empty}


\input{introduction.tex}
\input{preliminaries.tex}

\input{full_information.tex}
\input{limited_information.tex}

\input{simulations.tex}
\input{conclusions.tex}
\section*{Acknowledgements}
\noindent \small{This work was supported by the Swedish Foundation for Strategic Research through the CLAS project (grant RIT17-0046).}

\bibliographystyle{IEEEtran}
\bibliography{ref}
\input{appendix.tex}

\end{document}

%% file: introduction.tex

\begin{abstract}
This work investigates the problem of analyzing privacy of abrupt changes for general Markov processes. These processes may be affected by changes, or exogenous signals, that need to remain private. Privacy refers to the disclosure of information of these changes through observations of the underlying Markov chain. In contrast to previous work on privacy, we study the problem for an online sequence of data. We use theoretical tools from optimal detection theory to motivate a definition of online privacy based on the average amount of information per observation of the stochastic system in consideration. Two cases are considered: the full-information case, where the eavesdropper measures all but the signals that indicate a change, and the limited-information case, where the eavesdropper only measures the state of the Markov process. For both cases, we provide ways to derive privacy upper-bounds and compute policies that attain a higher privacy level. It turns out that the problem of computing privacy-aware policies is concave, and we conclude with some examples and numerical simulations for both cases.
\end{abstract}
\section{Introduction}\label{sec:introduction}
Being able to detect changes in stochastic systems  has several applications: it enables industrial quality control, fault detection, segmentation of signals,  monitoring in biomedicine, and more. The topic of change detection has been widely studied for nearly a century \cite{shewhart1931economic,veeravalli2014quickest,lai1998information,lai2010sequential,lorden1971procedures,moustakides1986optimal,page1954continuous,pollak1985optimal,shiryaev1963optimum,tartakovsky2014sequential}, and has recently sparked an interest in exploring the problem through the lens of differential privacy \cite{cummings2018differentially}.
Differential privacy \cite{dwork2014algorithmic} has emerged as a  technique for enabling data analysis while preventing information leakage. Privacy, in the context of linear dynamical systems, has been used to study the problem of private filtering \cite{le2013differentially}, the problem of  private parameters estimation \cite{wang2017differential}, and more.  Similarly, also private change-point detection algorithms  have been developed \cite{cummings2018differentially}, whose goal is to detect distributional changes at an unknown change-point in a sequence of data while making sure to satisfy a certain level of privacy.
 
In contrast to previous work on privacy, we study the scenario where an eavesdropper tries to detect a change in a controlled stochastic system $\mathcal{S}$. Eavesdropping, which is a leakage of information, leads to a loss of privacy. This privacy loss, in turn, may reveal private information regarding the system. For example, it may expose the action that a person performed on the system, or, in buildings, may reveal when a person enters or leaves an apartment. Furthermore, eavesdropping is more likely to happen if the system has many sensors, which is usually the case in modern cyber-physical systems. The impact of such an attack could be sensibly reduced, if not nullified, in case encryption is used. Nevertheless, encryption may not always be the best option due to increased processing time. Therefore, it is of paramount importance to be able to minimize information leakage  while at the same time satisfying some performance requirements of the system.  

Our analysis draws inspiration from \cite{alisic2020ensuring}, where the authors analyze the privacy properties of an autonomous linear system undergoing step changes. In contrast to their work, we consider the online case for generic Markov processes, whereas in \cite{alisic2020ensuring} they considered the case of offline change detection in linear systems.

\textit{Contributions: } the objectives of this work are twofold: (1) to properly define the problem of privacy in online change-detection problems for Markov processes; (2) to provide ways to derive privacy bounds and show how to compute policies that attain higher privacy level.  We conclude by providing: (A) a library to solve the optimization problems presented here and (B) an example for a linear dynamical system (more examples can be found in the library).

\textit{Organization of the paper: }  \cref{sec:preliminaries} introduces Quickest Change Detection, Markov Decision Processes and our proposed definition of privacy.  In \cref{sec:body}, we introduce the model; in  \cref{sec:full_information}  we analyze the case where the eavesdropper can measure both state and action, and in  \cref{sec:limited_information} we analyse  the case where only the state is measured. We conclude with examples and numerical results in \cref{sec:simulations}.

%% file: preliminaries.tex

\section{Preliminaries and Problem Formulation}\label{sec:preliminaries}
In this section we give a brief description of (1) Minimax Quickest Change Detection, (2) the framework of Markov Decision Processes and (3) the problem formulation.

\subsection{Minimax Quickest Change Detection (QCD)} 	Consider an agent willing to detect an abrupt change in a stochastic system. To this aim, the agent has access to a non-i.i.d. sequence of observations $\{Y_t\}_{t\geq 1}$. The change occurs at the unknown time $\nu$, and we let $\mathbb{P}_\nu$ denote the probability measure under which the system dynamics are generated if the change point is $\nu$. We also denote by $\mathbb{P}_\infty$ the probability measure in absence of a change point. Under $\mathbb{P}_\nu$, the conditional density function of $Y_t$ given $(Y_1,\dots, Y_{t-1})$ is $f_0(\cdot|Y_1,\dots, Y_{t-1})$ for $t<\nu$, and $f_1(\cdot|Y_1,\dots, Y_{t-1})$ for $t\geq\nu$. The agent needs to detect the change point in an online manner: her decision takes the form of a stopping time $T$ with respect to the filtration $\{ {\cal F}_t\} _{t\ge 1}$ where ${\cal F}_t=\sigma(Y_1,\ldots Y_t)$. In absence of any prior information about $\nu$, a common approach, due to Lordern and Pollak \cite{lorden1971procedures,pollak1985optimal}, is to aim at devising stopping rule $T$ minimizing the worst case expected delay\footnote{The essential supremum of a real-valued r.v. $X$ is defined up to an event with zero probability: $\esssup X =\inf\{a\in \mathbb{R} : \mathbb{P}_\nu[X\ge a]=0$\}.}
\begin{equation}
	\overline{\mathbb{E}}_1(T):=  \sup_{\nu\geq 1}\esssup \mathbb{E}_\nu[(T-\nu)^+| \set F_{\nu-1}],
\end{equation}
over all possible rules satisfying the constraint  $\mathbb{E}_\infty[T]\geq \bar T$ on the expected duration to false alarm (we impose this constraint since we work in a non-Bayesian setting, where it is not possible to impose a constraint on the  false alarm rate \cite{lai1998information}). Additionally,  for non-i.i.d. observation it is common to make an assumption on the convergence of the average log-likelihood ratio (see \cite{lai1998information} or \cite{tartakovsky2014sequential}), which permits us to find a lower bound on the expected delay.

\begin{assumption}\label{assump:lai_assumption} Define the log-likelihood ratio (LLR) as
 \begin{equation}\label{eq:llr}
 	Z_i \coloneqq \ln \frac{f_1(Y_i|Y_1,\dots, Y_{i-1})}{f_0(Y_i|Y_1,\dots, Y_{i-1})}.
 \end{equation}
Assume that $n^{-1}\sum_{t=\nu}^{\nu+n}Z_t$ converges a.s. under $\mathbb{P}_\nu$, to some constant $I$, and that, for all $\delta >0$,
\begin{small}
\begin{equation}
	\lim_{n\to\infty}\sup_{\nu\geq 1} \esssup \mathbb{P}_\nu \left(\max_{t\leq n} \sum_{i=\nu}^{t+\nu}Z_i\geq I(1+\delta)n\Big| \set F_{\nu-1}\right) = 0.
\end{equation}
\end{small}
\end{assumption}

Assumption \ref{assump:lai_assumption} involves conditioning on $\{Y_t\}_{t=1}^{\nu-1}$, and depends on $I$, which can be interpreted as the average amount of information per observation sample for discriminating between the two models $f_1$ and $f_0$ (note that 
Assumption \ref{assump:lai_assumption} is quite general, and holds, for example, for stable linear dynamical systems).
Under the above assumption, Lai \cite{lai1998information} established an asymptotic (as $\bar T\to\infty$) lower bound on the worst case expected delay of any stopping rule in $D({\bar T})$ (the set of rules satisfying $\mathbb{E}_\infty[T]\geq \bar T$):

\begin{equation}\label{eq:lower_bound_change_detection_nonidd} 
	\liminf_{\bar T\to\infty } \inf_{T \in D(\bar T)} \frac{\overline{\mathbb{E}}_1(T)}{\ln \bar T} \geq I^{-1}.
\end{equation}
This lower bound also provides an interpretation of $I$:  $I$ plays the same role in the change detection theory as the Cramer-Rao lower bound in estimation theory \cite{tartakovsky2014sequential}, hence it quantifies the detection difficulty.
It is also proved that the lower bound is achieved by the CUSUM algorithm with stopping time $T=\inf \left\{t: \max_{1\leq k\leq t} \sum_{i=k}^t Z_i \geq c \right\}$ provided that $c$ is chosen so that $\mathbb{E}_\infty[T]= \bar T$. These results can also be extended to unknown models through the generalized likelihood ratio test \cite{lai2010sequential}.

\subsection{Privacy and hardness of change detection inference} The asymptotic lower bound in \ref{eq:lower_bound_change_detection_nonidd} provides a notion of privacy  in online change-detection problems. In order to maintain privacy, we would like the statistical differences  before and after the abrupt change to be as small as possible. From the perspective of differential privacy \cite{dwork2008differential}, we are interested in bounding the following quantity $\sup_{\tau_N}\ln \frac{\mathbb{P}_\nu(\tau_N)}{\mathbb{P}_\infty(\tau_N)}$, where $\tau_N= (Y_1,\dots, Y_N)$ is a trajectory of size $N$. \begin{remark}In contrast to the classical definition of differential privacy, we are not interested in minimizing the statistical difference between two trajectories $(\tau,\tau')$, but the difference in any trajectory before and after the abrupt change.
\end{remark}

However, uniformly bounding $\frac{\mathbb{P}_\nu(\tau_N)}{\mathbb{P}_\infty(\tau_N)}$ may be detrimental. It is sensitive to outliers, and, in practice, results in unsatisfactory utility \cite{wang2016average}. Instead, a more natural approach is to bound $ \mathbb{E}_{\tau_N\sim {\cal D}}\left[\ln \frac{\mathbb{P}_\nu(\tau_N)}{\mathbb{P}_\infty(\tau_N)}\right]$ over some distribution ${\cal D}$. This quantity, also known as on-average KL-Privacy \cite{wang2016average}, is  distribution-specific quantity, and allows us to study the problem for a specific distribution ${\cal D}$. In this work it comes natural to choose ${\cal D} = \mathbb{P}_\nu$: if Assumption  \ref{assump:lai_assumption} is satisfied, and  we let $N\to\infty$,   we obtain that the on-average KL-privacy coincides with the  quantity $I$ in \cref{eq:lower_bound_change_detection_nonidd}.

This result is not surprising: $I$ dictates how difficult the detection problem is. As $I$ decreases, the time needed to discriminate between the two models increases, and thus becomes harder to note if an abrupt change happened. Therefore, the quantity $I$ lends itself well to define the privacy of an abrupt change.
\begin{definition}[Privacy of an abrupt change]
Consider the observations ${\cal Y}=\{Y_t\}_{t\geq 1}$ of a stochastic dynamical system, where the conditional density function of $Y_t$ given $(Y_1,\dots, Y_{t-1})$ is $f_0$ for $t<\nu$, and $f_1$ otherwise. If $\mathcal{Y}$ satisfies Assumption \ref{assump:lai_assumption}, we define the privacy level of $\mathcal{Y}$ as $\mathcal{I}(\mathcal{Y}) = I^{-1}$.
\end{definition}

In controlled system, we can modify the control policy to manipulate $I$, and, in turn, create a trade-off between control performance and information leakage. We can select a policy that increases the privacy (i.e., minimizes $I$), but this may come at the expense of decreased utility.


%

\subsection{Markov Decision Processes (MDPs)} 
We study stochastic systems that can be modeled using the MDP framework.
An MDP $M$ is a controlled Markov chain, described by a tuple $M=({\cal X}, {\cal U}, P, r)$, where ${\cal X}$ and ${\cal U}$ are the state and action spaces, respectively. $P: {\cal X}\times {\cal U} \to \Delta({\cal X})$ denotes the conditional state transition probability distributions ($\Delta({\cal X})$ denote the set of distributions over ${\cal X}$), i.e., $P(x'|x,u)$ is the probability to move from state $x$ to state $x'$ given that action $u$ is selected. Finally, $r: {\cal X}\times {\cal U}\to \mathbb{R}$ is the reward function. A (randomized) control policy $\pi: {\cal X}\to \Delta({\cal U})$ determines the selected actions, and $\pi(u|x)$ denotes the probability of choosing $u$ in state $x$ under $\pi$. For simplicity, we focus on ergodic MDPs, where ${\cal X}$ and ${\cal U}$ are finite, and where any policy $\pi$ generates a positive recurrent Markov chain with stationary distribution $\mu^\pi$. The value of a policy $\pi$ is defined as $V_M^\pi = \lim_{N\to \infty} \mathbb{E}_M^\pi\left[\frac{1}{N} \sum_{t=1}^{N} r(X_t, U_t)\right]$ (here $U_t$ is distributed as $\pi(\cdot|X_t)$). In ergodic MDPs, the objective is to find a policy $\pi$ with maximal value $ V_M^\star  = \max_\pi V_M^\pi$. In the sequel, we denote by $D(P,Q)=\mathbb{E}_{\omega\sim P}[\ln\frac{P}{Q}(\omega)]$ the KL-divergence between two distributions $P$ and $Q$, and by $d(p,q) = p\ln\frac{p}{q} + (1-p)\ln\frac{1-p}{1-q}$ the  KL-divergence between two Bernoulli distributions of parameter $p$ and $q$.
For two probability measures $P$ and $Q$ we write $P\ll Q$ if $P$ is absolutely continuous with respect to $Q$, i.e., for every measurable set $A$, $Q(A) = 0\Rightarrow P(A)=0$.

\subsection{Problem formulation}\label{sec:body}

In this paper, we investigate the utility privacy trade-off in controlled dynamical systems with one change point. One can also extend the analysis to multiple change points but for simplicity of the exposition, we restrict our attention to a single change point, always denoted by $\nu$. We formulate the problem for ergodic MDPs with finite state and action spaces (however, our results hold for other types of system, e.g., classical linear systems). Consider two ergodic MDPs $M_0$ and $M_1$, and assume that the main agent faces $M_0$  before $\nu$ and $M_1$ after. Let $M_i = ({\cal X},{\cal U}, P_i, r_i), i=0,1$, and assume that $P_1$ is absolute continuous w.r.t. $P_0$, which means that for all pair $(x,u)$, $P_1(x,u)\ll P_0(x,u)$.\\
\noindent We make the following assumptions for the two agents:
\begin{itemize}
	\item  The {\bf main agent} {\it knows} the time at which the MDP changes, and applies the control policy $\pi_0$ (resp. $\pi_1$) for $t<\nu$ (resp. $t\ge \nu$). We assume that just before the change occurs, the system state distribution is $\mu_0^{\pi_0}$, the stationary distribution of the Markov chain induced in $M_0$ by $\pi_0$ (resp. $\mu_1^{\pi_1}$ is the stationary distribution induced by $\pi_1$ on $M_1$).
	\item The {\bf eavesdropper} wishes to infer the change point $\nu$ by observing the system's dynamics.
\end{itemize}
\noindent Then, based on what the eavesdropper can observe, we consider two possible scenarios (depicted in \cref{fig:mdp_scheme}):
\begin{enumerate}
	\item The \textit{full information} scenario, where the eavesdropper is able to observe $Y_t:=(X_t,U_t)$ at time $t$.
	\item The \textit{limited information} case, where the eavesdropper is able to observe only $Y_t:=X_t$ at time $t$.
\end{enumerate}
In the full information case, we denote by $I_F(\pi_0,\pi_1)$ the inverse of the privacy level. Similarly, $I_L(\pi_0,\pi_1)$ is the inverse of the privacy level in the limited information scenario. We will prove that these levels are well-defined (in the sense that Assumption 1 holds). The objective of the main agent is to design the control policies $\pi_0$ and $\pi_1$ realizing an appropriate trade-off between their rewards and privacy level. The utility of $(\pi_0,\pi_1)$ is a linear combination, parametrized by $\rho\in [0,1]$, of the ergodic rewards before after the change point: $V(\rho,\pi_0,\pi_1):= \rho V_{M_0}^{\pi_0} +(1-\rho) V_{M_1}^{\pi_1}$. To assess the trade-off between utility and privacy of the main agent, we will analyze the solution of the following optimization problem for different values of $\lambda\ge 0$:
\begin{equation}\label{eq:optpb}
\sup_{\pi_0,\pi_1} \rho V_{M_0}^{\pi_0} +(1-\rho) V_{M_1}^{\pi_1} -\lambda I(\pi_0,\pi_1),
\end{equation}
where $I(\pi_0,\pi_1)= I_F(\pi_0,\pi_1)$ (resp. $=I_L(\pi_0,\pi_1)$) in the full (resp. limited) information scenario.

\begin{figure}[t]
	\centering
	\includegraphics[width=0.8\columnwidth]{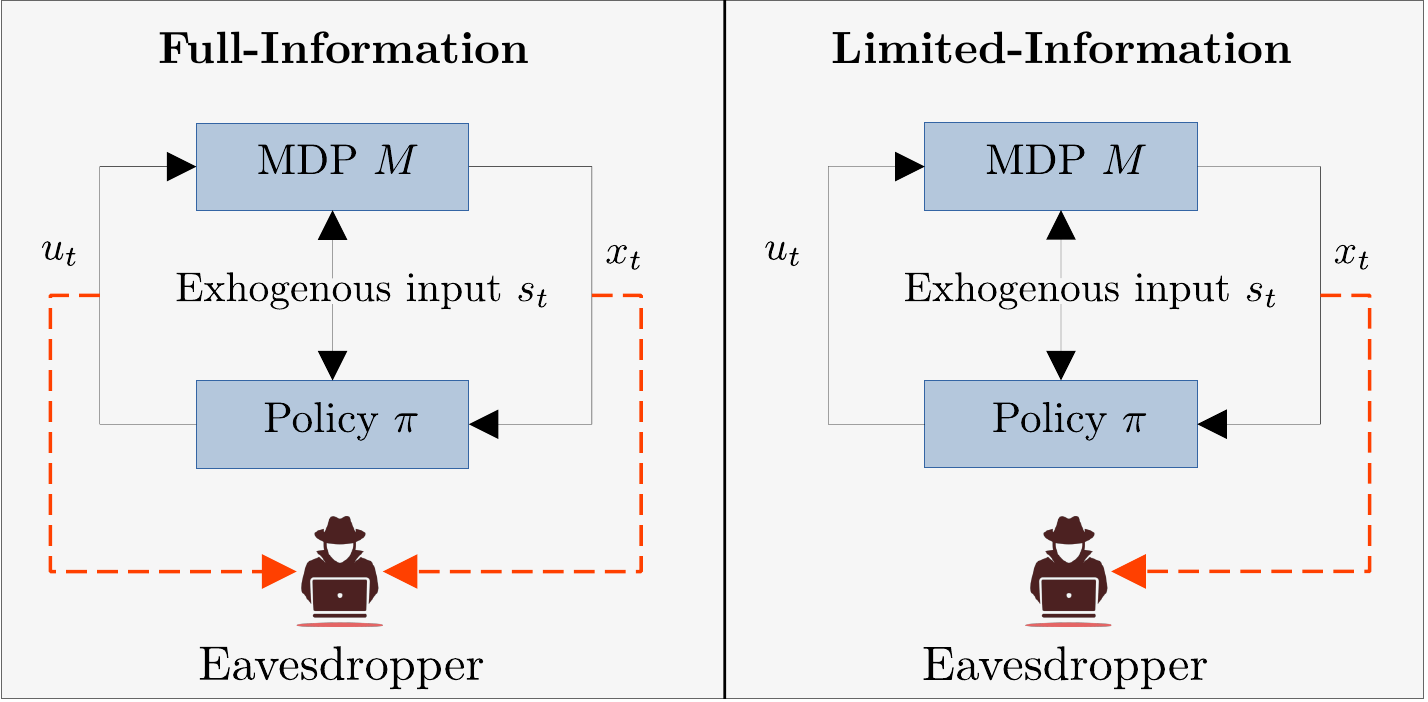}
	\caption{Scenarios considered}
	\label{fig:mdp_scheme}
\end{figure}

%% file: full_information.tex

\section{Full-information scenario} \label{sec:full_information}
In the full-information case, the eavesdropper can measure both the state and action $(X_t,U_t)$ at time $t$. We first analyze the privacy level $I_F(\pi_0,\pi_1)$, and then investigate the utility-privacy trade-off in this case.
\subsection{Privacy level}
In the full-information case, $I_F(\pi_0,\pi_1)$ can be decomposed in the sum of the average KL-divergence of the two models and the KL-divergence of the two policies:
\begin{theorem}\label{theorem:privacy_full_info}
(i) If for all $x\in \supp(\mu_1^{\pi_1})$, $\pi_1(x) \ll \pi_0(x)$, then the sequence of observations  $\{Y_t\}_{t\ge1}$ (made by the eavesdropper), with $Y_t=(X_t,U_t)$, satisfies Assumption 1, and we have: 	
\begin{align}\label{eq:I_full_information}
I_F(\pi_0,\pi_1) = &\ \mathbb{E}_{x\sim \mu_1^{\pi_1}, u\sim \pi_1(x)}\left[D(P_1(x,u),P_0(x,u))\right]\nonumber \\
& \ \ \ \ + \mathbb{E}_{x\sim \mu_1^{\pi_1}}\left[D(\pi_1(x), \pi_0(x))\right]. 
\end{align}
(ii) If $\exists x\in \supp(\mu_1^{\pi_1}): \pi_1(x) \nll \pi_0(x)$  then $I_F=\infty$.
\end{theorem}
\begin{proof}
 If 	$\pi_1(x) \ll \pi_0(x)$ does not hold for some $x\in \supp(\mu_1^{\pi_1})$, then Assumption 1 does not hold, and (ii) follows by definition.  To prove (i), using the Markov property, one easily get an expression of the conditional densities $f_0$ and $f_1$, and deduce what is $Z_i$ in (\ref{eq:llr}): for all $i\ge 2$,
\[
Z_{i} = \ln \frac{f_1(Y_i|Y_{i-1})}{f_0(Y_i|Y_{i-1})}=\ln \frac{\pi_1(U_i|X_i)P_1(X_i|X_{i-1}, U_{i-1})}{\pi_0(U_i|X_i)P_0(X_i|X_{i-1}, U_{i-1})},
\]
By ergodicity,  it follows that $n^{-1}\sum_{t=\nu}^{\nu+n} Z_t$ converges to
\[
\resizebox{\hsize}{!}{%
$I_F(\pi_0,\pi_1)=\sum_{(x,u)\in  {\cal X}\times {\cal U}} \mathbb{E}[Z_2|X_{1}=x, U_1=u]  \pi_1(u|x)\mu_1^{\pi_1}(x)$}.
\]
	
Furthermore, $\mathbb{E}[Z_2|X_{1}=x, U_1=u] $ is
	\begin{align*}
		&\sum_{(y,a)\in {\cal X}\times {\cal U}} \ln \frac{\pi_1(a|y)P_1(y|x,u)}{\pi_0(a|y)P_0(y|x,u)}\pi_1(a|y)P_1(y|x,u),\\
	   =&\underbrace{\sum_y\ln \frac{P_1(y|x,u)}{P_0(y|x,u)}P_1(y|x,u)}_{=D(P_1(x,u),P_0(x,u))} \\
	     &\qquad + \sum_{y} P_1(y|x,u) \underbrace{\sum_a\ln \frac{\pi_1(a|y)}{\pi_0(a|y)}\pi_1(a|y)}_{=D(\pi_1(y), \pi_0(y))}.
	\end{align*}
	Observe now that $\sum_{x,u,y} D(\pi_1(y), \pi_0(y))P_1(y|x,u) \mu_1^{\pi_1}(x)$ is equal to
	\begin{align*}
	&\sum_{x,y} D(\pi_1(y), \pi_0(y))P_1^{\pi_1}(y|x)\mu_1^{\pi_1}(x)\\
	=&
	\sum_{y} D(\pi_1(y), \pi_0(y)) \underbrace{\sum_{x}P_1^{\pi_1}(y|x)\mu_1^{\pi_1}(x)}_{=\mu_1^{\pi_1}(y)}.
	\end{align*}
	Then, the result follows from the fact that  for a stationary distribution $\mu_1^{\pi_1}$ it holds that $\sum_{x}P_1^{\pi_1}(y|x)\mu_1^{\pi_1}(x)=\mu_1^{\pi_1}(y)$.
\end{proof}
Theorem \ref{theorem:privacy_full_info}, as well as the other theorems and propositions in this paper, can be established for general state-action spaces, and therefore it is quite general. A first important consequence of Theorem \ref{theorem:privacy_full_info} is that when $\pi_0$ and $\pi_1$ are different deterministic policies, then the absolute continuity condition is not met and  the level of privacy is 0 (since the actions reveal the change point to the eavesdropper). Hence, there is a price to pay to get a non-zero level of privacy. 
Theorem \ref{theorem:privacy_full_info} also allows us to compute the policies  maximizing the level of privacy (or equivalently minimizing $I_F(\pi_0,\pi_1)$):
\begin{proposition}\label{proposition:lower_bound_privay}
The best level of privacy in the full-information case is given by $\underline{I}_F = \inf_{\pi_0,\pi_1} I_F(\pi_0,\pi_1)$ that can be computed by solving the following linear program
\begin{equation}
\begin{aligned}\label{eq:min_c1_full_information}
\min_{\xi \in\Delta({\cal X}\times {\cal U})} \quad  &\sum_{x,u} \xi_{x,u}  D(P_1(x,u),P_0(x,u))\\ \textrm{s.t.} \quad  &\sum_{u} \xi_{*,u}^\top P_1(u) =\sum_{u} \xi_{*,u}^\top
\end{aligned}
\end{equation}
where $P_1(u)$ is a $|{\cal X}|\times |{\cal X}|$ matrix containing the transition probabilities for  action $u$ in MDP $M_1$. The policies achieving $\underline{I}_F$ are given by $\pi_1(u|x) = \xi_{x,u}/\|\xi_{x,*}\|_1$ and $\pi_0=\pi_1$.
\end{proposition}
\begin{proof}
Observe that for any $\pi_1$ the infimum of $I_F(\pi_0,\pi_1)$ over $\pi_0$ is simply $\pi_0=\pi_1$. Therefore the problem becomes to minimize $\mathbb{E}_{u\sim \pi_1(x), x\sim \mu_1^{\pi_1}}\left[D(P_1(x,u),P_0(x,u))\right]$ over $\pi_1$.  Let $\xi \in \Delta({\cal X}\times {\cal U})$ be a distribution over the states and the actions. We can equivalently rewrite $\mathbb{E}_{u\sim \pi_1(x), x\sim \mu_1^{\pi_1}}\left[D(P_1(x,u),P_0(x,u))\right]$ through a change of variables $\xi_{x,u} = \pi_1(u|x)\mu_1^{\pi_1}(x)$, subject to the affine constraint $\sum_{u} \xi_{*,u}^\top P_1(u) =\sum_{u} \xi_{*,u}^\top$ that guarantees stationarity of the distribution. The result follows from this rewriting.
\end{proof} 
\noindent Alternatively, it is possible to compute the best level of privacy  $\underline{I}_F$ by solving an MDP $({\cal X},{\cal U},P_1,r)$ with  reward function $r(x,u)=-D(P_1(x,u), P_0(x,u))$.

\subsection{Privacy-utility trade-off} 
Next we investigate the utility-privacy trade-off by studying the solution of the optimization problem (\ref{eq:optpb}) for different values of $\lambda$. We denote the objective function by:
\begin{equation}
 V_F(\rho,\lambda,\pi_0,\pi_1)=    V(\rho,\pi_0,\pi_1) - \lambda I_F(\pi_0,\pi_1).
\end{equation}


Note that we may be interested in optimizing just $\pi_1$, the policy after the change, i.e., solve $\sup_{\pi_1} V_{M_1}(\pi_1) -\lambda I_F(\pi_0,\pi_1)$ for some fixed $\pi_0$ (where $\pi_0$ may be the optimal policy in $M_0$ for example). This problem corresponds to $\rho=0$ in (\ref{eq:optpb}), and hence is  just a special case in our analysis. In the following theorem, we show that solving the problem is equivalent to minimizing a difference of convex functions under convex constraints, and is hence a concave minimization problem.

\begin{theorem} \label{theorem:utility_privacy_problem_full_information}
The solution to $\sup_{\pi_0,\pi_1} V_F(\rho,\lambda,\pi_0,\pi_1)$ is obtained by solving:
\begin{equation}\label{eq:utility_privacy_problem_full_information}
\begin{aligned}
\min_{(\gamma,\xi^0,\xi^1)\in \Omega} \quad & \gamma-  \lambda\sum_{x} \|\xi_{x,*}^1\|_1\ln\frac{\|\xi_{x,*}^1\|_1}{\|\xi_{x,*}^0\|_1} \\
\textrm{s.t.} \quad & \sum_u (\xi_{*,u}^i)^\top P_i(u) =\sum_{u} (\xi_{*,u}^i)^\top \quad i=0,1\\
& \sum_{x,u}\lambda f(x,u,\xi^0,\xi^1) - q(x,u,\rho,\xi^0, \xi^1)\leq \gamma
\end{aligned}
\end{equation}
where $\Omega =\mathbb{R}\times \Delta({\cal X}\times {\cal U})\times \Delta({\cal X}\times {\cal U})$, and
\begin{small}
\begin{align*}
f(x,u,\xi^0,\xi^1)&=\xi_{x,u}^1D(P_1(x,u),P_0(x,u))+\xi_{x,u}^1\ln \frac{\xi_{x,u}^1}{\xi_{x,u}^0},\\
q(x,u,\rho,\xi^0, \xi^1)&=\rho\xi_{x,u}^1 r_1(x,u) +(1-\rho)\xi_{x,u}^0 r_0(x,u),
\end{align*}
\end{small}
and by choosing $\pi_i(u|x) = \xi_{x,u}^i/\|\xi_{x,*}^i\|_1$ for $i=0,1$.
\end{theorem}
\begin{proof}
Observe that the problem is equivalent to
$
\min_{\pi_0,\pi_1} -\rho V_{M_1}(\pi_1)-(1-\rho)V_{M_0}(\pi_0) + \lambda I_F(\pi_0,\pi_1)
$. Through a change of variable  $\xi_{x,a}^i=\pi_i(a|x)\mu_i^{\pi_i}(x)$, as in Proposition \ref{proposition:lower_bound_privay}, the problem becomes
\begin{equation*}
\resizebox{\hsize}{!}{%
$
\begin{aligned}
\min_{\xi^0,\xi^1} \quad & \sum_{x,a} -\rho \xi_{x,a}^1r_1(x,a) - (1-\rho) \xi_{x,a}^0r_0(x,a) +\lambda I_F(\pi_0,\pi_1)\\
\textrm{s.t.} \quad & \sum_{a} (\xi_{*,a}^i)^\top P_i(a) =\sum_{a} (\xi_{*,a}^i)^\top \quad i=0,1.
\end{aligned}$
}
\end{equation*}
Note now that $\mathbb{E}_{x\sim \mu_1^{\pi_1}}\left[D(\pi_1(x), \pi_0(x))\right]$ in $I_F$ is equivalent to 
$\sum_{x,u}\xi_{x,u}^1 \left[\ln \frac{\xi_{x,u}^1}{\xi_{x,a}^0} - \ln\frac{\|\xi_{x,*}^1\|_1}{\|\xi_{x,*}^0\|_1} \right],$
that is the difference of two convex functions. Consequently, the original objective is a difference of convex functions. Define $f,g$ as in the statement of the theorem. The problem can rewritten as a concave program with convex constraint by introducing an additional parameter $\gamma\in \mathbb{R}$, with constraint $\sum_{x,u}\lambda f(x,u,\xi^0,\xi^1) - q(x,u,\rho,\xi^0, \xi^1)\leq \gamma$\ifdefined\shortpaper .\else
.
\fi
\end{proof}
Problem (\ref{eq:utility_privacy_problem_full_information}) can be solved using 
methods from DC programming (Difference of Convex functions). Note, however, that there are specific instances of (\ref{eq:utility_privacy_problem_full_information}) that could be convex. This happens when $D(\pi_1(x),\pi_0(x))$ is constant for all $x$, 
or if we impose the additional constraint $\pi_0=\pi_1$. The latter constraint appears if $\rho=1$, in which case the problem is equivalent to solving an MDP  $({\cal X}, {\cal U},P_1, r_1^\lambda)$ with modified reward $r_1^\lambda(x,u)= r_1(x,u)-\lambda D(P_1(x,u), P_0(x,u))$.

We have a few additional remarks to make regarding Theorem \ref{theorem:utility_privacy_problem_full_information}. The term $-  \sum_{x} \|\xi_{x,*}^1\|_1\ln\frac{\|\xi_{x,*}^1\|_1}{\|\xi_{x,*}^0\|_1}$ can be interpreted as the negative KL-divergence between the two stationary distributions $-D(\mu_1^{\pi_1}, \mu_0^{\pi_0})$. This term causes the problem to be concave. Solutions of (\ref{eq:utility_privacy_problem_full_information})  favor distributions $\xi^0,\xi^1$ that are close to each other in the KL-divergence sense. As a consequence, in case $r_0=r_1$,  the solutions of (\ref{eq:utility_privacy_problem_full_information}) will hardly depend on $\rho$.  To see this, let $\delta_{x,u} \coloneqq \xi^{1}_{x,u}-\xi^{0}_{x,u}$ and notice that the following equality holds $\rho V_{M_1}(\pi_1) + (1-\rho)V_{M_0}(\pi_0)= \sum_{x,u} r_0(x,u) (\rho \delta_{x,u} + \xi_{x,u}^0)$. The KL-divergence is an upper bound of the total variation distance. It follows that a small KL-divergence between $\xi_1$ and $\xi_0$ implies a small value of $\delta_{x,u}$ in the absolute sense for all $(x,u)$, and thus a small dependence on $\rho$. 

%% file: limited_information.tex
\section{Limited-information case}\label{sec:limited_information}

We now analyze the limited-information case, where the eavesdropper has access to the states $\{X_t\}_{t\ge 1}$ only. 

\subsection{Privacy level}

As in the full-information case, we can  characterize $I_L$. Theorem \ref{theorem:privacy_full_info}. Unfortunately, it is not possible to have a separation of the KL-divergences between the models and the policies as in the full-information case.
\begin{theorem}\label{theorem:privacy_limited_info}
(i) If for all $x \in \supp(\mu_1^{\pi_1})$, $P_1^{\pi_1}(x)\ll P_0^{\pi_0}(x)$, then the sequence of observations $\{Y_t\}_{t\ge 1}$ , with $Y_t=X_t$, satisfies Assumption 1, and we have: 
\begin{equation}\label{eq:I_limited_information}
I_L(\pi_0,\pi_1) = \mathbb{E}_{x\sim \mu_1^{\pi_1}}\left[D\left(P_1^{\pi_1}(x), P_0^{\pi_0}(x)\right)\right].
\end{equation}
 Furthermore, $I_L(\pi_0,\pi_1) \leq I_F(\pi_0,\pi_1)$. 
 \begin{equation} \label{eq:lower_bound_IL}
 \mathbb{E}_{x\sim \mu_1^{\pi_1}}\left[\sup_{x'}d\left(P_1^{\pi_1}(x'|x), P_0^{\pi_0}(x'|x)\right)\right] \leq  I_L(\pi_0,\pi_1).
 \end{equation}
 (ii) If $\exists x\in \supp(\mu_1^{\pi_1}): P_1^{\pi_1}(x)\nll P_0^{\pi_0}(x)$  then $I_L=\infty$.
\end{theorem}
\begin{proof}
\ifdefined\shortpaper The proof of $I_L= \mathbb{E}_{x\sim \mu_1^{\pi_1}}\left[D\left(P_1^{\pi_1}(x), P_0^{\pi_0}(x)\right)\right]$ is omitted for simplicity. The former inequality $I_L(\pi_0,\pi_1)\leq I_F(\pi_0,\pi_1)$ follows from an application of the log-sum inequality. The latter  is a  consequence of the  fundamental data processing inequality \cite{garivier2019explore}, where one has $D\left(P_1^{\pi_1}(x), P_0^{\pi_0}(x)\right) \geq d\left(\mathbb{E}_{P_1^{\pi_1}(x)}[Z], \mathbb{E}_{P_0^{\pi_0}(x)}[Z]\right)$ for a measurable random variable $Z$. By choosing $Z$ as the event of transitioning from $x$ to $x'$, and optimizing over $x'$, concludes the proof.
\else
We prove (i) and the bounds of $I_L(\pi_0,\pi_1)$. For $Y_t= (X_t)$ eq. \ref{eq:llr} becomes
	$Z_i = \ln \frac{P_1^{\pi_1}(X_i|X_{i-1})}{P_0^{\pi_0}(X_i|X_{i-1})}$
	for $i\in \{2,\dots, t\}$. The limit $\lim_{n\to\infty}n^{-1}\sum_{t=\nu}^{t+\nu}Z_t$ converges to
	\begin{align*}
	I_L(\pi_0,\pi_1)&=\sum_{x\in {\cal X}} \mu_1^{\pi_1}(x) \sum_{y\in {\cal X}}P_1^{\pi_1}(y|x) \ln \frac{P_1^{\pi_1}(y|x)}{P_0^{\pi_0}(y|x)}.
	\end{align*}
	where the inner term is just the KL-divergence between $P_1^{\pi_1}(\cdot|x)$ and $P_1^{\pi_1}(\cdot|x)$, thus 	$I_1(\pi_1,\pi_0) = \mathbb{E}_{x\sim \mu_1^{\pi_1}}\left[D\left(P_1^{\pi_1}(x), P_0^{\pi_0}(x)\right)\right].$
	To prove  the inequality just apply the log-sum inequality on the inner term in $I_L$. 
	\begin{align*}\sum_{y\in {\cal X}}P_1^{\pi_1}(y|x) \ln \frac{P_1^{\pi_1}(y|x)}{P_0^{\pi_0}(y|x)}&\leq\\
 \sum_{y,u} P_1(y|x,u)&\pi_1(u|x)\ln \frac{P_1(y|x,u)\pi_1(u|x)}{P_0(y|x,u)\pi_0(u|x)}.
	 \end{align*} Compare now the new expression with the one in theorem \ref{theorem:privacy_full_info} to see that it is equal to $I_F(\pi_0,\pi_1)$. 
	 The last inequality  is a  consequence of the fundamental data processing inequality \cite{garivier2019explore}, where one has $D\left(P_1^{\pi_1}(x), P_0^{\pi_0}(x)\right) \geq d\left(\mathbb{E}_{P_1^{\pi_1}(x)}[Z], \mathbb{E}_{P_0^{\pi_0}(x)}[Z]\right)$ for a measurable random variable $Z$. By choosing $Z$ as the event of transitioning from $x$ to $x'$, and optimizing over $x'$, concludes the proof.
 \fi
\end{proof}
\noindent Note that since we assume that for all $(x,u)$, $P_1(x,u)\ll P_0(x,u)$, the condition to get a finite $I_L(\pi_0,\pi_1)$ holds if $\pi_1\ll \pi_0$ (but this is not a necessary condition). In addition, as expected, the limited information case yields a higher privacy level than the full information scenario. Further observe that the lower bound in (\ref{eq:lower_bound_IL}) is tighter than $\min_x D\left(P_1^{\pi_1}(x), P_0^{\pi_0}(x)\right)$, and can be used to upper bound the privacy level $I_L^{-1}$.  
However, computing policies  that attain the best level of achievable privacy is more challenging compared to the full-information case. The fact that it is not possible to separate the policies and the models in Theorem \ref{theorem:privacy_limited_info} as we did in Theorem \ref{theorem:privacy_full_info} implies that we cannot use the trick to optimize only over $\pi_1$ to find the best level of privacy.   As a consequence, it turns out that  finding the best level of achievable privacy becomes a concave problem, in general. 
\begin{proposition}\label{proposition:lower_bound_privay_limited_information}
The best level of privacy in the limited-information case is given by $\underline{I}_L = \inf_{\pi_0,\pi_1} I_L(\pi_0,\pi_1)$ that can be computed by solving the following concave program
\begin{equation}\label{eq:12}
	\begin{aligned}
		&\min_{(\gamma,\alpha,\xi^1)\in\Omega'} \quad  \gamma -\sum_{x}\|\xi_{x,*}^1\|_1\ln \|\xi_{x,*}^1\|_1\\
		&\textrm{s.t.} \quad  \sum_{u} (\xi_{*,u}^1)^\top P_1(u) =\sum_{u} (\xi_{*,u}^1)^\top,\\
		& \qquad \sum_{x,y} \left(\sum_{u} P_1(y|x,u)\xi_{x,u}^1\right) \ln \frac{\sum_{u'} P_1(y|x,u')\xi_{x,u'}^1}{\sum_{u'} P_0(y|x,u')\alpha_{x,u'}}\leq \gamma.
	\end{aligned}
\end{equation}
where $\Omega'=\mathbb{R}\times \Delta({\cal U})^{{\cal X}}\times \Delta({\cal X}\times {\cal U})$. 
\end{proposition}
\begin{proof}
Similarly to the full-information case we perform a change of variable so that the problem becomes a minimization over state-action distributions. \begin{align*}
	I_L(\pi_0,\pi_1)&=\sum_{x\in {\cal X}} \mu_1^{\pi_1}(x) \sum_{y\in {\cal X}}P_1^{\pi_1}(y|x) \ln \frac{P_1^{\pi_1}(y|x)}{P_0^{\pi_0}(y|x)}.
\end{align*}
Let $\xi_{x,u}^1 = \pi_1(u|x)\mu_{1}^{\pi_1}(x)$, and denote the policy $\pi_0$ by $\alpha \in \Delta({\cal U})^{{\cal X}}$. Thus $I_L(\pi_0,\pi_1)$ is equivalent to
\begin{small}
\begin{align*}
	&=\sum_{x,y} \left(\sum_{u} P_1(y|x,u)\xi_{x,u}^1\right) \ln \frac{\sum_{u'} P_1(y|x,u')\xi_{x,u'}^1}{\|\xi_{x,*}^1\|_1\sum_{u'} P_0(y|x,u')\alpha_{x,u'}}\\
	&=\underbrace{-\sum_{x,y} \left(\sum_{u} P_1(y|x,u)\xi_{x,u}^1\right) \ln \|\xi_{x,*}^1\|_1}_{(u)}\\
	&\qquad+ \underbrace{\sum_{x,y} \left(\sum_{u} P_1(y|x,u)\xi_{x,u}^1\right) \ln \frac{\sum_{u'} P_1(y|x,u')\xi_{x,u'}^1}{\sum_{u'} P_0(y|x,u')\alpha_{x,u'}}}_{(b)}.
\end{align*}\end{small}
Note that (a) is equal to $- \sum_{x}\|\xi_{x,*}^1\|_1\ln \|\xi_{x,*}^1\|_1$. One can conclude that the expression is a difference of convex functions. Consequently it is possible to use the same approach as in Theorem \ref{theorem:utility_privacy_problem_full_information}
\ifdefined\shortpaper to get the result.
\else
. We can rewrite the problem as (\ref{eq:12}) by introducing  an additional variable $\gamma \in \mathbb{R}$.
\fi
\end{proof}
\noindent As already mentioned (\ref{eq:12}) may be hard to solve, but there are still some instances where it corresponds to a convex program. This is the case if $D(P_1^{\pi_1}(x), P_0^{\pi_0}(x))$ does not depend on $x$. Alternatively, consider the inequality $I_L(\pi_0, \pi_1) \leq \max_x D\left(P_1^{\pi_1}(x), P_0^{\pi_0}(x)\right)$. Minimizing the right-hand side over $(\pi_0,\pi_1)$ is a convex problem, and can be used as an approximation to $\inf_{\pi_0,\pi_1} I_L(\pi_0,\pi_1)$. As a final remark, note that contrarily to Proposition \ref{proposition:lower_bound_privay}, it is not necessarily true that at the infimum of $I_L(\pi_0,\pi_1)$, the two policies coincide.
\subsection{Privacy-utility trade-off} 
\noindent We end this section by providing a way to compute policies that maximize utility and privacy in the limited-information case. The concave program to be solved is, for the most part, similar to the one solved in Theorem \ref{theorem:utility_privacy_problem_full_information}, with the only difference being the privacy term that appears in the constraint. 
\begin{theorem} \label{theorem:utility_privacy_problem_limited_information}
	Let $\rho \in[0,1], \lambda\geq 0$ and define 
	$
	V_L(\rho,\lambda,\pi_0,\pi_1)=   V(\rho,\pi_0,\pi_1) - \lambda I_L(\pi_0,\pi_1)$. The solution to $\sup_{\pi_0,\pi_1} V_L(\rho,\lambda,\pi_0,\pi_1)$ is obtained by solving
	\begin{equation}\label{eq:utility_privacy_problem_limited_information}
	\resizebox{\hsize}{!}{%
		$
	\begin{aligned}
	\min_{(\gamma,\xi^0,\xi^1)\in \Omega} \quad & \gamma-  \lambda\sum_{x}\|\xi_{x,*}^1\|_1\ln \frac{\|\xi_{x,*}^1\|_1}{\|\xi_{x,*}^0\|_1}\\
	\textrm{s.t.} \quad & \sum_{u} (\xi_{*,u}^i)^\top P_i(u) =\sum_{u} (\xi_{*,u}^i)^\top \quad i=0,1\\
	& \sum_{x}\left(\lambda f(x,\xi_0,\xi_1) -\sum_{u}q(x,u,\rho,\xi^0,\xi^1) \right)\leq \gamma
	\end{aligned}$}
	\end{equation}
	where $\Omega $ and $q$ are as in Theorem \ref{theorem:utility_privacy_problem_full_information}  and \[
	\resizebox{\hsize}{!}{%
	$
	f(x,\xi_0,\xi_1)=\sum_y \left(\sum_{u} P_1(y|x,u)\xi_{x,u}^1\right) \ln \frac{\sum_{u'} P_1(y|x,u')\xi_{x,u'}^1}{\sum_{u'} P_0(y|x,u')\xi_{x,u'}^0},$}
	\]
	and by choosing $\pi_i(u|x) = \xi_{x,u}^i/\|\xi_{x,*}^i\|_1$ for $i=0,1$.
\end{theorem}
\begin{proof}
The proof is along the same lines as that of Theorem \ref{theorem:utility_privacy_problem_full_information} by making use of the decomposition of $I_L$ shown in Proposition \ref{proposition:lower_bound_privay_limited_information}.
\end{proof}

%% file: simulations.tex

\section{Examples and numerical results}\label{sec:simulations}
 \noindent  We implemented a library\footnote{The code and instructions to run the simulations can be found at \url{https://github.com/rssalessio/PrivacyStochasticSystems}.} built on top of the DCCP library  \cite{shen2016disciplined} to  solve the concave problems presented above. Due to space constraints, we restrict our attention to specific  linear systems and a simple MDP with three states.

\subsection{Additive changes in linear systems}
\noindent  Consider the following linear model $x_{t+1} = Ax_t + Bu_t + F\theta \indi_{\{t\geq \nu\}} + w_t$ where $x_t\in \mathbb{R}^n$ is the state, $u_t\in \mathbb{R}^m$ is the control signal, and $w_t\in \mathbb{R}^n$ is a white noise sequence with $0$ mean and covariance $Q$. The parameter $\theta \in \mathbb{R}^k$ models the exogenous input, unknown to the eavesdropper.\\

\noindent\textbf{Full information.} In this case, the best level of privacy is obtained with $\pi_0=\pi_1$, and we can prove that it does not depend on $\pi_0$. This is a simple consequence of Proposition 1 and the fact that 
$
D(P_1(x,u), P_0(x,u))=(1/2)\theta^\top F^\top Q^{-1} F\theta=\underline{I}_F.
$
In turn, the privacy level depends solely on the signal-to-noise ratio (SNR) $\frac{1}{2}\theta^\top F^\top Q^{-1} F\theta$, which increases as the minimum eigenvalue of $Q$ increases. This result agrees with the conclusions of \cite{alisic2020ensuring}.

 Next, to investigate the privacy-utility trade-off, we assume the columns of $B$ are linearly independent and that there exists  $K \in \mathbb{R}^{n\times m}$ such that $A+BK$ is a Schur matrix. We assume the reward function is  $r(x,u) = -x^\top x$. To shorten the notation, we use the following definitions: $L\coloneqq I-A-BK$, $E\coloneqq  (L^{-1})^{\top}L^{-1}$.
\begin{proposition}\label{proposition:full_info_utility_privacy_linear_case}
Suppose the control laws are of the type $u_t = Kx_t+ \beta_t^0$ for $t<\nu$, and $u_t=Kx_t+\beta_t^1$ otherwise, where $\beta_t^i$ is i.i.d. white Gaussian noise distributed according to $\mathcal{N}(\alpha_i, R)$, with $R \succ 0$. Then, the utility-privacy value function is 
	\begin{align*}
		V_F(\rho, \lambda,\alpha_0,\alpha_1)&=-(1-\rho)\alpha_0^\top B^\top EB\alpha_0 -\tr(\Sigma)\\
		&\quad - \rho(B\alpha_1+F\theta)^\top E(B\alpha_1+F\theta)\\
		&\quad -\lambda \frac{c_\theta+ (\alpha_1-\alpha_0)^\top R^{-1}(\alpha_1-\alpha_0)}{2},
	\end{align*}
	where $c_\theta=\theta^\top F^\top Q^{-1} F\theta$ and $\Sigma$ satisfies the Lyapunov equation $\Sigma = (A+BK)\Sigma (A+BK)^\top + Q+BRB^\top$. For $\lambda >0, \rho\in[0,1]$, the solutions to $\max_{\alpha_0,\alpha_1} V_F(\rho,\lambda,\alpha_0,\alpha_1)$ are given by
	\begin{align*}
		\alpha_0^\star(\rho,\lambda) &= -\rho\left(B^\top T(\rho,\lambda )B \right)^{-1}B^\top E F\theta,\\
		\alpha_1^\star(\rho,\lambda)&=\left( 2\frac{1-\rho}{\lambda} RB^\top EB +I\right)\alpha_0^\star(\rho,\lambda),
	\end{align*}
	where $T(\rho,\lambda) =  \left(I  + \frac{2(1-\rho)\rho}{\lambda} EBRB^\top\right)E$. Moreover, the solution to $\max_{\alpha_1} V_F(1,\lambda,\alpha_1, 0)$ is given by $\alpha_1^\star(\lambda) = -\left(RB^\top E B + \frac{\lambda}{2}I\right)^{-1}RB^\top E F\theta$.
\end{proposition}
\ifdefined\shortpaper The proof is omitted for simplicity.
\else
\begin{proof}
	If $A+BK$ is Schur, then the system converges to a stationary distribution before and after the change. Specifically, the two distributions are $\mu_0 \sim \mathcal{N}(L^{-1}B\alpha_0, \Sigma)$ and $\mu_1\sim \mathcal{N}(L^{-1}(B\alpha_1+F\theta), \Sigma)$, where $\Sigma$  satisfies the Riccati equations
	$
	\Sigma = (A+BK)\Sigma (A+BK)^\top + Q+BRB^\top
	$.
	Therefore the value of the policy before and after the change is
	$V_{M_0}(\alpha_0) = -\mathbb{E}_{x\sim\mu_0}[x^\top x]$ and  $V_{M_1}(\alpha_1) = -\mathbb{E}_{x\sim\mu_1}[x^\top x]$. For a normal random variable $y\sim \mathcal{N}(m,W)$ it holds that $\mathbb{E}[y^\top y] = m^\top m+\tr(W)$. Therefore we have
	$V_{M_0}(\alpha_0) = -\alpha_0^\top B^\top EB\alpha_0-\tr(\Sigma)$
	and
	$ V_{M_1}(\alpha_1)=-(B\alpha_1+F\theta)^\top E (B\alpha_1+F\theta)-\tr(\Sigma)$.
	The information value on the other hand is
	\[I_F = \frac{\overbrace{\theta^\top F^\top Q^{-1} F\theta}^{c_\theta} + (\alpha_1-\alpha_0)^\top R^{-1}(\alpha_1-\alpha_0)}{2}. \]
	Then  $-V_F(\rho, \lambda,\alpha_0,\alpha_1)$ is convex, and maximizing $V_F$ is equivalent to minimizing $-V_F$. Taking the gradient of $V_F$ with respect to $\alpha_0,\alpha_1$ yields
	\begin{align*}
			-\nabla_{\alpha_{0}} V_F&=2(1-\rho) B^\top EB\alpha_0 -\lambda R^{-1}(\alpha_1-\alpha_0)\\
			-\nabla_{\alpha_{1}} V_F&=2\rho B^\top E(B\alpha_1+F\theta)+\lambda R^{-1}(\alpha_1-\alpha_0).
	\end{align*}
	$\rho=0$ implies $\alpha_0=\alpha_1$ and $\alpha_0=0$ since $B^\top EB$ is full rank. $\rho=1$, similarly, implies $\alpha_0=\alpha_1$ and $B^\top EB\alpha_1=-B^\top EF\theta$, hence $\alpha_1 = -(B^\top EB)^{-1}B^\top EF\theta$. For the general case using the first equation one can write the following expression for $\alpha_1$
	\begin{align*}
		&\left( 2\frac{1-\rho}{\lambda} RB^\top EB +I\right)\alpha_0 =\alpha_1.
	\end{align*}
	Now consider $-(\nabla_{\alpha_{1}} V_F+\nabla_{\alpha_{0}} V_F)=0$ and plug in the expression found for $\alpha_1$
	\begin{align*}
		&(1-\rho) B^\top EB\alpha_0 \\
		&\qquad +\rho B^\top E\left(B\left( 2\frac{1-\rho}{\lambda} RB^\top EB +I\right)\alpha_0+F\theta\right)=0
	\end{align*}
	that is also equal to
	\begin{align*}
			(1-\rho) B^\top EB\alpha_0 +\rho B^\top EB&\left( 2\frac{1-\rho}{\lambda} RB^\top EB +I\right)\alpha_0\\
			&\qquad =-\rho B^\top E F\theta.
	\end{align*}
	Then, we can conclude that the left-hand side is equal to
	\[B^\top \left[(1-\rho) E  + \frac{2(1-\rho)\rho}{\lambda} EBRB^\top E  + \rho E\right]B\alpha_0.\]
	Let now $T(\rho,\lambda) =  \left(I  + \frac{2(1-\rho)\rho}{\lambda} EBRB^\top\right)E $, hence
	\[ \alpha_0 = -\rho\left(B^\top T(\rho,\lambda )B \right)^{-1}B^\top E F\theta,\]
	from which follows also the expression for $\alpha_1$. 
	Finally, notice that the solution to $\max_{\alpha_1} V_F(1,\lambda,\alpha_1, 0)$ can be easily derived by using the equation $-\nabla_{\alpha_1} V_F(1,\lambda,\alpha_1, 0)=2B^\top E (B\alpha_1+F\theta)+\lambda R^{-1}\alpha_1=0$.
\end{proof}\fi
\noindent  Proposition \ref{proposition:full_info_utility_privacy_linear_case} uses stochastic policies to ensure absolute continuity of the policies. Consequently, one can optimize over the mean of the policy while keeping fixed the covariance term $R$. The larger the eigenvalues of $R$, the better it is in terms of privacy (but worse performance).\\
\begin{figure}[t]
	\centering
	\includegraphics[width=\columnwidth]{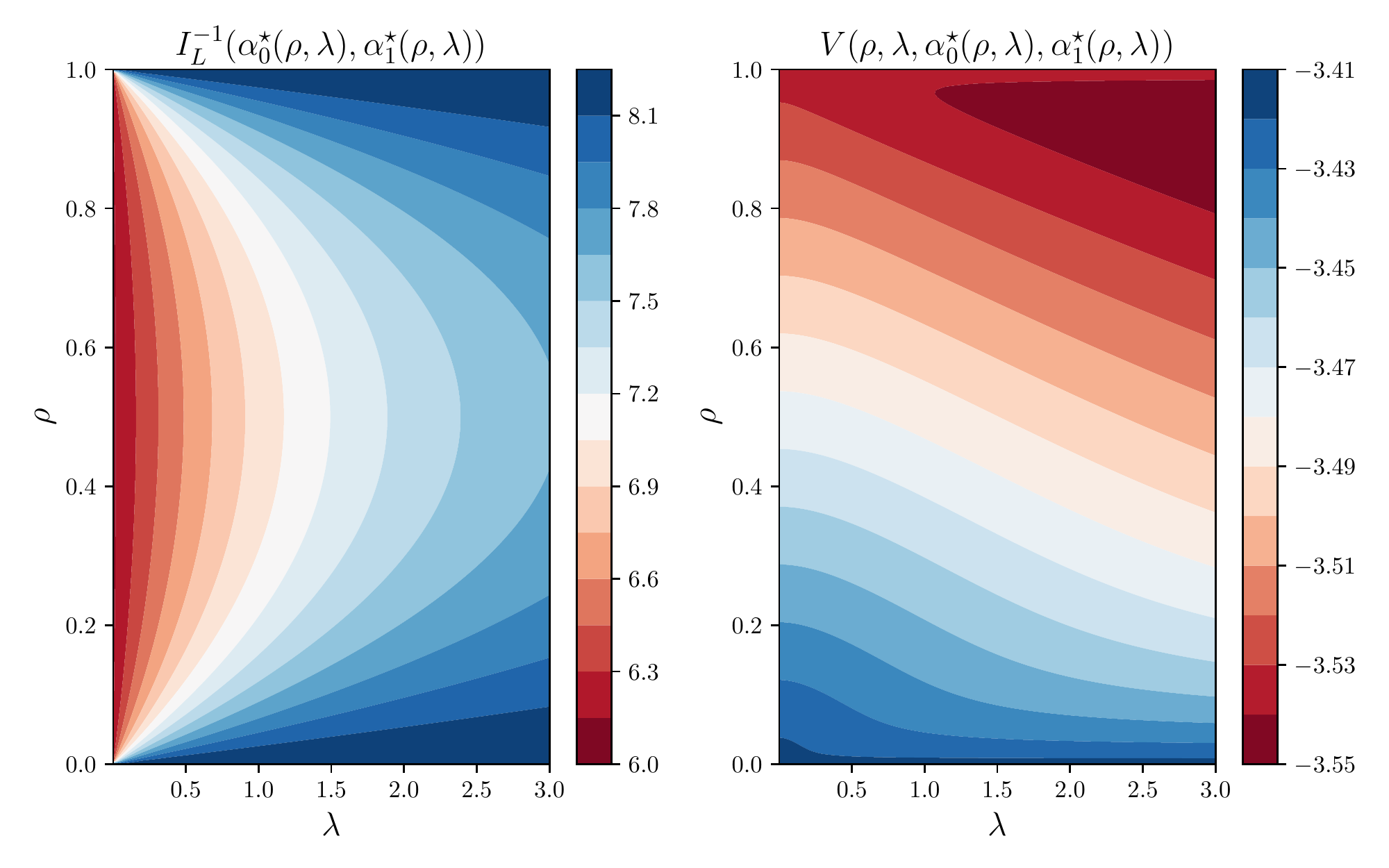}
	\caption{Limited information case (example in \ref{eq:linear_sys_example}): plot of $V$ and $I_L^{-1}$ as function of the optimal solutions $\alpha_0^\star, \alpha_1^\star$. The larger the values (in blue), the better.}
	\label{fig:linear_sys_limited_case}
\end{figure}

\noindent \textbf{Limited information.}  To find the best privacy level in the limited information case, we exploit the presence of process noise to just consider deterministic policies. 
Let the policies be described by $u_t = g_0(x_t) + \Delta g(x_t)$ for $t\geq\nu$ and $u_t=g_0(x_t)$ for $t<\nu$, where $g_0,\Delta g$  are deterministic mappings from $\mathbb{R}^n\to \mathbb{R}^m$.
\ifdefined\shortpaper
\else
Therefore it follows that the two densities are  $P_1^{\pi_1}(x'|x)=\mathcal{N}(Ax+Bg_1(x)+F\theta, Q)$ and $P_0^{\pi_0}=\mathcal{N}(Ax+Bg_0(x), Q)$.
Consequently, we obtain that the KL-divergence is
\begin{equation*}
\int_{\mathbb{R}^n}P_1^{\pi_1}(x'|x)\ln \frac{P_1^{\pi_1}(x'|x)}{P_0^{\pi_0}(x'|x)} \textrm{d}x'=\frac{1}{2}h(\theta,x)^\top Q^{-1} h(\theta, x),
\end{equation*}
where $h(\theta,x) = F\theta + B\Delta g(x).$ 
\fi
One easily deduce that the infimum of $I_L$ is attained for $\Delta g =- (B^\top Q B)^{-1}B^\top QF\theta$, which means that at the minimum, the difference in the control laws does not depend on $x$, and that the minimum is attained for a control law $g_1(x)$ that cancels out the effect of the additive change. Hence $\underline{I}_L = \frac{1}{2}\theta^\top F^\top G^\top Q^{-1}GF\theta$ where $G=I- B(B^\top QB)^{-1}B^\top Q$.

 Next, we investigate the privacy-utility trade-off. As previously mentioned, we consider deterministic policies.
 Define $\tilde B_T(M) \coloneqq ( B^\top (M + T) B)^{-1}B^\top M$ for any symmetric invertible matrix $M\in \mathbb{R}^{n\times n}$ and symmetric semi-positive definite matrix $T$. Then, we have  the following result.
\begin{proposition}\label{proposition:limited_info_utility_privacy_linear_case}
Consider the limited-information case.  Consider deterministic control laws of the type $u_t=Kx_t+\alpha_0$ for $t<\nu$ and $u_t=Kx_t+\alpha_1$ for $t\geq \nu$. The utility-privacy value function $V_L$ is
\begin{align*}
	V_L(&\rho,\lambda,\alpha_0,\alpha_1)=-(1-\rho)\alpha_0^\top B^\top EB\alpha_0 -\tr(\Sigma)\\
	& - \rho(B\alpha_1+F\theta)^\top E(B\alpha_1+F\theta)\\
	& -\lambda \frac{1}{2}\left(F\theta + B(\alpha_1-\alpha_0)\right)^\top Q^{-1} \left(F\theta + B(\alpha_1-\alpha_0)\right),
\end{align*}
  where $\Sigma$ satisfies  $\Sigma = (A+BK)\Sigma (A+BK)^\top + Q$.  For $\lambda>0, \rho\in[0,1],$ the solutions $\alpha_0^\star(\rho,\lambda)$ and $\alpha_1^\star(\rho,\lambda)$ are 
	\begin{align*}
	\alpha_1^\star(\rho,\lambda) &=-\tilde{B}_0\left( (1-\rho) EB\tilde{B}_{2(1-\rho) E/\lambda}(Q^{-1}) +\rho E\right)F \theta,\\
	\alpha_0^\star(\rho,\lambda) &=  \tilde{B}_{2(1-\rho) E/\lambda}(Q^{-1})(F\theta+ B\alpha_1^\star(\rho,\lambda))
\end{align*}
that simplify to $\alpha_0=0,\alpha_1=-\tilde{B}_0(Q^{-1})F\theta$ if $\rho=0$, $\alpha_0=\tilde{B}_0(Q^{-1})(I -B\tilde{B}_0(E))F\theta,\alpha_1=-\tilde{B}_0(E)F\theta$ if $\rho=1$. Moreover, the solution to $\max_{\alpha_1} V_L(1,\lambda,\alpha_1, 0)$ is given by $\alpha_1^\star(\lambda) = -\tilde{B}_0\left(\frac{2}{\lambda}E+Q^{-1}\right)F\theta$.
\end{proposition}
\ifdefined\shortpaper The proof is omitted for simplicity.
\else
\begin{proof}
	The first part of the proof is identical to the one in Proposition \ref{proposition:full_info_utility_privacy_linear_case}. Now, one can  find out that  $I_F$ is
	\[I_F = \frac{\left(F\theta + B(\alpha_1-\alpha_0)\right)^\top Q^{-1} \left(F\theta + B(\alpha_1-\alpha_0)\right)}{2}. \]
	 Taking the gradient of $V_F$ with respect to $\alpha_0,\alpha_1$ yields
	\begin{align*}
		-\nabla_{\alpha_{0}} V_F&=2(1-\rho) B^\top EB\alpha_0\\
		&\qquad -\lambda B^\top Q^{-1} \left(F\theta + B(\alpha_1-\alpha_0)\right),\\
		-\nabla_{\alpha_{1}} V_F&=2\rho B^\top E(B\alpha_1+F\theta)\\
		&\qquad+\lambda B^\top Q^{-1} \left(F\theta + B(\alpha_1-\alpha_0)\right).
	\end{align*}
	Therefore, it is possible to conclude that for $\rho=0$ the solution is given by $\alpha_0=0$ since $B^\top E B$ is full rank, and $\alpha_1= -(B^\top Q^{-1} B)^{-1}B^\top Q^{-1} F\theta = -\tilde{B}_0(Q^{-1})F\theta$. 
	Similarly, for $\rho=1$ one has
	\begin{equation*}
		\begin{cases}
			-\lambda B^\top Q^{-1} \left(F\theta + B(\alpha_1-\alpha_0)\right)=0,\\
			2B^\top E(B\alpha_1+F\theta)+\lambda B^\top Q^{-1} \left(F\theta + B(\alpha_1-\alpha_0)\right)=0.
		\end{cases}
	\end{equation*}
	Using the first equation in the second one  concludes that $2B^\top E(B\alpha_1+F\theta)=0$ and hence $\alpha_1=- \tilde{B}_0(E)F\theta$. From the first equation one obtains $ B^\top Q^{-1}B \alpha_0= B^\top Q^{-1} \left(F\theta + B\alpha_1\right)$ and consequently $\alpha_0 = \tilde{B}_0(Q^{-1})(F\theta+B\alpha_1)=\tilde{B}_0(Q^{-1})(I-B \tilde{B}_0(E))F\theta$.
	For the general case using $\nabla_{\alpha_{0}} V_F=0$ one can write
	\[
	B^\top\left( 2\frac{1-\rho}{\lambda} E +  Q^{-1} \right)B\alpha_0= B^\top Q^{-1}(F\theta +  B\alpha_1)
	\]
	that results in $\alpha_0 =\tilde{B}_{2(1-\rho) E /\lambda}(Q^{-1})(F\theta +  B\alpha_1).$  Replacing this expression in $-(\nabla_{\alpha_{0}} V_F+\nabla_{\alpha_{1}} V_F)=0$ gives
	\[((1-\rho) B^\top E B\tilde{B}_{2(1-\rho) E/\lambda}(Q^{-1}) +\rho B^\top E)(B\alpha_1+F\theta)=0 \]
	that is
	\begin{align*}
		B^\top E((1-\rho) &B\tilde{B}_{2(1-\rho) E/\lambda}(Q^{-1}) +\rho I)B\alpha_1=\\
		&- B^\top E( (1-\rho) B\tilde{B}_{2(1-\rho)/\lambda}(Q^{-1}) +\rho I)F\theta,
	\end{align*}
	and, consequently,
	\begin{align*}
		\alpha_1 &=-\tilde{B}_0\left( (1-\rho) EB\tilde{B}_{2(1-\rho) E/\lambda}(Q^{-1}) +\rho E\right)F \theta.
	\end{align*}
Finally, the solution to $\max_{\alpha_1} V_L(1,\lambda,\alpha_1, 0)$ can be easily derived by using the equation $-\nabla_{\alpha_1} V_L(1,\lambda,\alpha_1, 0)=2B^\top E (B\alpha_1+F\theta)+\lambda B^\top Q^{-1}(B\alpha_1+F\theta)=0$.
\end{proof}\fi
\noindent Observe that both the value and the information term contain $B\alpha_1+F\theta$. This suggests choosing $\alpha_0=0$, and minimizing the impact of $F\theta$ using an appropriate $\alpha_1$. This choice of $(\alpha_0,\alpha_1)$ corresponds to the case $\rho=0$.
This observation is  confirmed by numerical results to have a better performance. 

\medskip
\noindent\textbf{Numerical example.} Here we consider the linear system in the limited-case scenario, with parameters 
\begin{equation}\label{eq:linear_sys_example}
	Q=I_2, A=\begin{bmatrix}0 & 1\\1 & 1\end{bmatrix}, B=\begin{bmatrix}0.01\\1\end{bmatrix}, F=\begin{bmatrix}0.5 \\ 0.7\end{bmatrix}
\end{equation}
and $\theta=1$. The control law  is $u_t=Kx_t + \alpha_t$ where $\alpha_t=\alpha_0$ for $t<\nu$ and $\alpha_t=\alpha_1$ for $t\geq \nu$. The control gain stabilizes the system, chosen as $K=\begin{bmatrix}-0.7 & -0.9\end{bmatrix}$.
In Fig. \ref{fig:linear_sys_limited_case} are shown results for the privacy-utility trade-off as a function of $(\rho, \lambda)$. Notice that for $\rho=1$ we obtain the best result, as previously observed in the discussion of Proposition \ref{proposition:full_info_utility_privacy_linear_case}.
\ifdefined\shortpaper\else
In Fig. \ref{fig:linear_sys_limited_case_example} is shown the average value of $\|x\|_2^2$, computed over $10^3$ simulations, with $95\%$ confidence interval (grayed area).
\begin{figure}[t]
	\centering
	\includegraphics[width=\columnwidth]{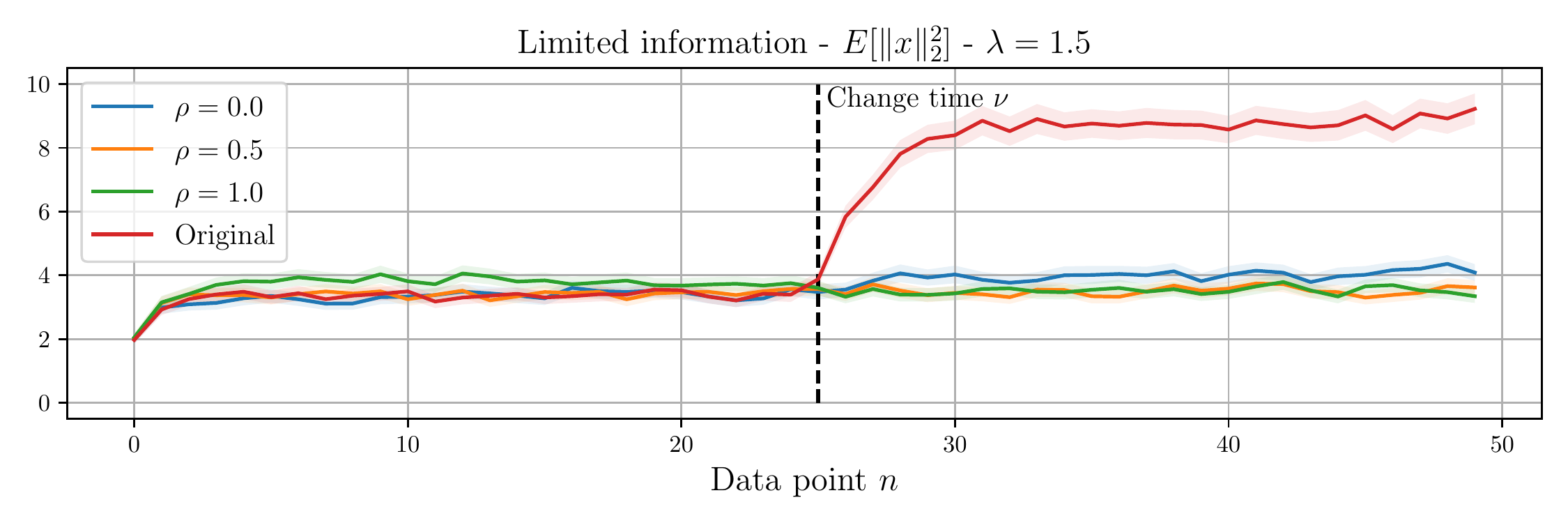}
	\caption{Limited information case (example in \ref{eq:linear_sys_example}): plot of the average value of $\|x\|_2^2$ for different values of $\rho$ and $\lambda=1.5$. The grayed area depicts the confidence interval (95\%).}
	\label{fig:linear_sys_limited_case_example}
\end{figure}

\subsection{3-States MDP}
We illustrate our results in an  MDP with $3$ states and $2$ actions (the details can be found in the code).
 The densities are as follows
\begin{small}
\begin{alignat*}{2}
&P_0(a_1) = \begin{bmatrix} 
.6 &.3 &.1\\
.05 & .85 & .1\\
.15 & .15 &.7
\end{bmatrix},\quad
&&P_0(a_2) = \begin{bmatrix} 
.5 &.2 &.3\\
.5 & .3 & .2\\
.3 & .3 &.4
\end{bmatrix}\\
&P_1(a_1) = \begin{bmatrix} 
.3 &.3 &.4\\
.35 & .5 & .15\\
.8 & .05 &.15
\end{bmatrix},\quad
&&P_1(a_2) = \begin{bmatrix} 
.3 &.55 &.15\\
.8 & .1 & .1\\
.5 & .3 &.2
\end{bmatrix}.
\end{alignat*}\end{small}
\begin{figure}[b]
	\centering
	\includegraphics[width=0.75\columnwidth]{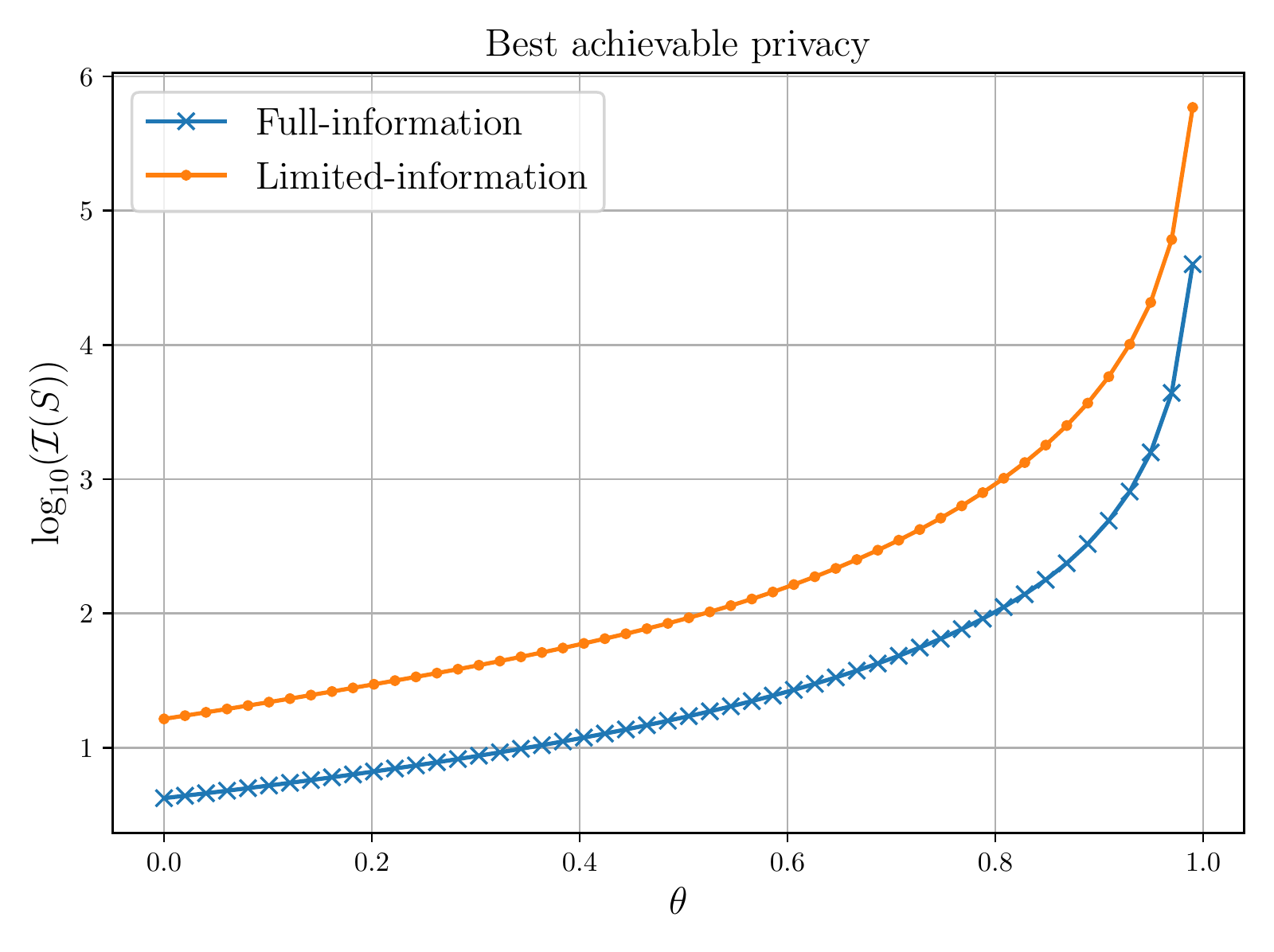}
	\caption{Scaling of the best achievable privacy  between $P_0$ and $P_\theta$ as a function of $\theta$ for both the full and limited information cases. Notice that the $y$-scale is logarithmic.}
	\label{fig:sim_mdp3states_best_privacy}
\end{figure}


Using this example, we analyze how privacy changes according to how "similar" the two models are. For that purpose, we examine what is the best level of privacy between $P_0$ and $P_\theta$, where $P_\theta(y|x,a) = \theta P_0(y|x,a) + (1-\theta) P_1(y|x,a)$, and let $\theta$ range between $0$ and $1$. Results are shown in Fig. \ref{fig:sim_mdp3states_best_privacy}. As one may expect, for $\theta=1$ the privacy level tends to $\infty$, since the two models coincide. For $\theta=0$ we have the level of privacy between $P_0$ and $P_1$. 

\fi

%% file: conclusions.tex
\section{Conclusions}
\noindent In this work, we analyzed the problem of minimizing information leakage of abrupt changes in Markov Decision Processes. By computing policies that minimize the statistical difference between the system before and after the change, one can reduce the loss of privacy resulting from this leakage of information. Future work will focus on removing the assumption that the agent perfectly knows when the change occurs, and how can Reinforcement Learning be applied to compute policies that minimize information leakage.


%% file: appendix.tex
\newpage
\section*{Appendix}

In this section we shall see that the function
\[q(\alpha, \beta) = \sum_{x,a} \alpha_{x,a} \log \frac{\alpha_{x,a}/ \sum_{a'}\alpha_{x,a'}}{\beta_{x,a}/ \sum_{a'}\beta_{x,a'}},\quad \alpha,\beta \in \Delta(X\times U) \]
is not necessarily convex. One can prove  that $q$ is equal to
\[q(\alpha, \beta) = \underbrace{\sum_{x,a} \alpha_{x,a} \log \frac{\alpha_{x,a}}{\beta_{x,a}}}_{f(\alpha,\beta)} - \underbrace{\sum_{x} \|\alpha_{x,*}\|_1 \log \frac{\|\alpha_{x,*}\|_1}{\|\beta_{x,*}\|_1}}_{g(\alpha,\beta)}. \]
For a convex set $\mathcal{X}$ a function $h: \mathcal{X} \to \mathbb{R}$  is convex if the following condition holds $\forall\lambda \in [0,1]$ and $x,y\in \mathcal{X}$:
\[D_h(x,y,\lambda) = \lambda h(x) + (1-\lambda) h(y) - h(\lambda x + (1-\lambda)y)\geq 0.\]
Let $\mathcal{X}= \Delta(X\times U)$ and $x=(\alpha, \beta), y=(\alpha', \beta')$. Then, $D_q(x,y,\lambda)\geq 0$ is equivalent to $D(f,x,y,\lambda)- D_g(x,y,\lambda)\geq 0$. Since $f(x) -g(x)\geq 0$ it is not necessarily true that $D_f(x,y,\lambda)- D_g(x,y,\lambda)\geq 0$. 
For example, consider $|X|=|U|=2$. Then, the following values
\[\alpha= \begin{bmatrix}
.5704 & .0206\\
.1980 & .2110\\
\end{bmatrix},
\beta= \begin{bmatrix}
.1312 & .1403\\
.3757 & .3529\\
\end{bmatrix}
\]
\[
\alpha'= \begin{bmatrix}
.2891 & .0753\\
.5033 & .1322\\
\end{bmatrix},
\beta'= \begin{bmatrix}
.1031 & .3591\\
.3672 & .1706\\
\end{bmatrix}
\]
yield $D_q(x,y,\lambda)\leq 0$ for all $\lambda \in [0,1]$ (Fig. \ref{fig:q_example_1}).
\begin{figure}[h]
\centering
\includegraphics[width=\columnwidth]{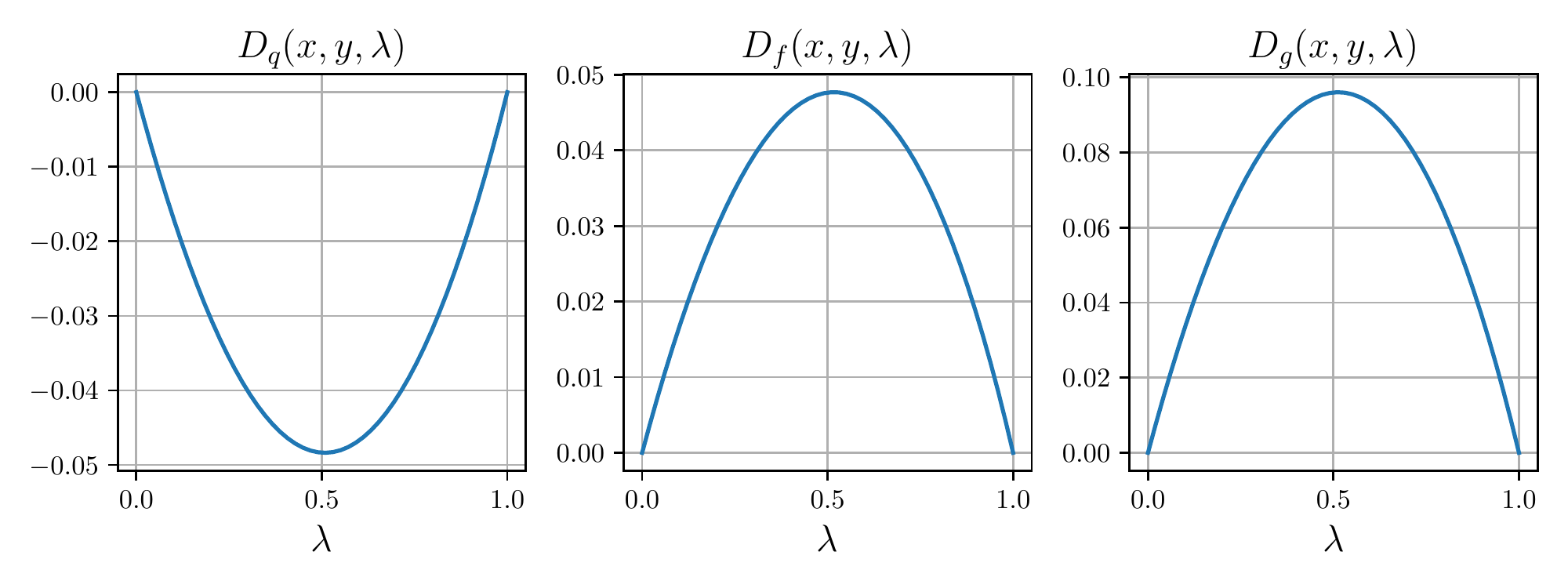}
\caption{Example where $D_q$ is lower than $0$ for all $\lambda\geq 0$.}
\label{fig:q_example_1}
\end{figure}
\\Another example is shown in Fig. \ref{fig:q_example_2}, where $(\alpha,\beta)$ are
\[\alpha= \begin{bmatrix}
.2110 & .3764\\
.3246 & .0881\\
\end{bmatrix},
\beta= \begin{bmatrix}
.4428 & .3469\\
.0297 & .1805\\
\end{bmatrix}
\]
and $(\alpha',\beta')$
\[
\alpha'= \begin{bmatrix}
.1935 & .3282\\
.4342 & .0441\\
\end{bmatrix},
\beta'= \begin{bmatrix}
.3474 & .2314\\
.0416 & .3796\\
\end{bmatrix}.
\]
\begin{figure}[h]
\centering
\includegraphics[width=\columnwidth]{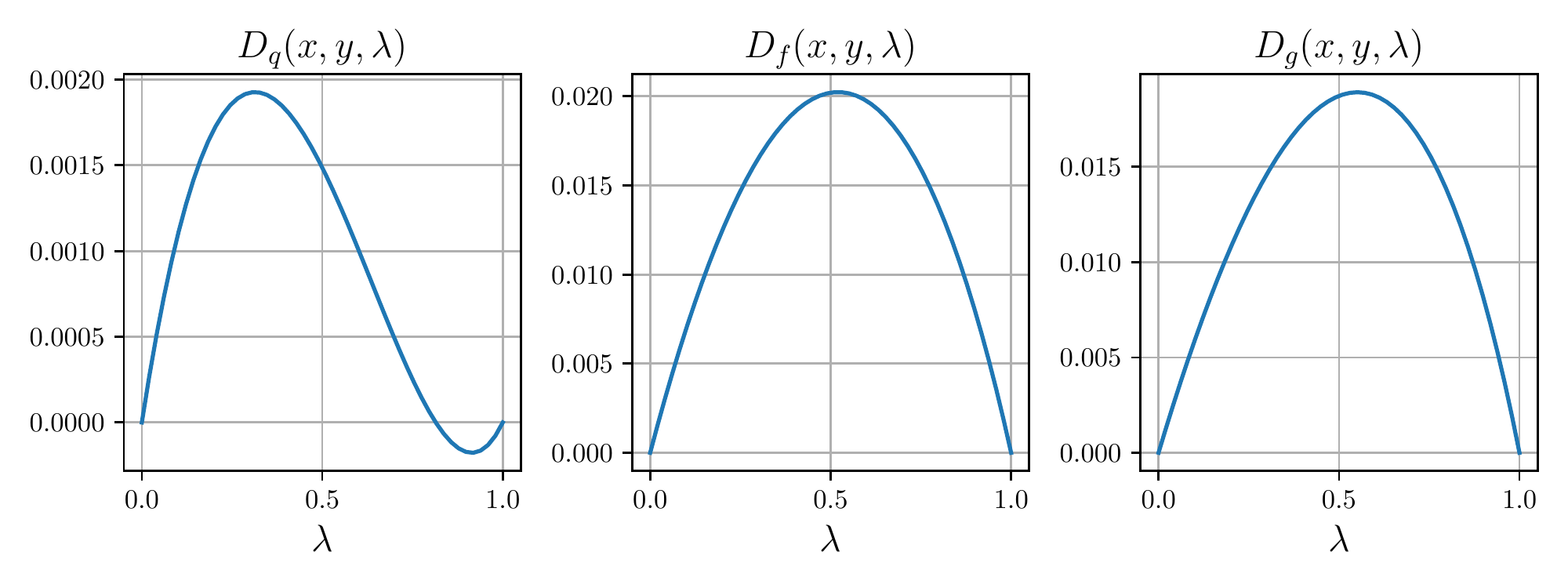}
\caption{Example where $D_q$ can be either positive or negative.}
\label{fig:q_example_2}
\end{figure}